\documentclass[11pt]{article}

\usepackage{todonotes}
\usepackage{amsthm}
\usepackage{amsmath}
\usepackage{amsfonts}
\usepackage{amssymb}
\usepackage{algorithm2e}
\usepackage{xcolor}
\usepackage{lipsum}
\usepackage[margin=1in]{geometry}
\usepackage[numbers,sort]{natbib}
\usepackage{hyperref}

\usepackage{tikz}
\usetikzlibrary{decorations.pathreplacing}
\usetikzlibrary{decorations.pathmorphing}
\tikzset{snake it/.style={decorate, decoration=snake}}
 
 \usepackage{float}
 
 \newtheorem{definition}{Definition}
 \newtheorem{lemma}{Lemma}
 \newtheorem{proposition}[]{Proposition}
  \newtheorem{observation}[]{Observation}
 \newtheorem{theorem}[]{Theorem}

 \newcommand{\dist}{\texttt{dist}}
 
 \newcommand{\Oish}{\widetilde{O}}

\title{Simple Linear-Size Additive Emulators\footnote{This work was supported by NSF:AF 2153680.}}
\author{
\begin{tabular}{c}
    Gary Hoppenworth  \\
   University of Michigan EECS   \\
    \texttt{garytho@umich.edu} \\
\end{tabular}
}
\date{}

\begin{document}

\maketitle

\begin{abstract}
Given an input graph $G = (V, E)$, an additive emulator $H = (V, E', w)$ is a sparse weighted graph that preserves all distances in $G$ with small additive error. A recent line of inquiry has sought to determine the best additive error achievable in the sparsest setting, when $H$ has a linear number of edges. In particular, the  work of [Kogan and Parter, ICALP 2023], following [Pettie, ICALP 2007],  constructed linear size emulators with $+O(n^{0.222})$ additive error. It is known that the worst-case additive error must be at least $+\Omega(n^{2/29})$  due to [Lu, Vassilevska Williams, Wein, and Xu, SODA 2022].

We present a simple linear-size emulator construction that achieves additive error $+O(n^{0.191})$. 
Our approach extends the path-buying framework developed by [Baswana,  Kavitha,  Mehlhorn, and  Pettie, SODA 2005] and [Vassilevska Williams and Bodwin, SODA 2016] to the setting of sparse additive emulators. 


\end{abstract}

\section{Introduction}
Spanners and emulators are well-studied graph objects which aim to approximately preserve distances in the input graph metric $G$, while reducing the number of edges in the graph representation.
In particular, a spanner $H$ of an input graph $G$ is a sparse subgraph that  approximately preserves distances in $G$. Emulators are a natural generalization of spanners that allow $H$ to be any weighted graph on the same vertex set as $G$. Spanners and emulators have applications in many areas of computer science, including fast graph algorithms \cite{ chen2022maximum, dalirrooyfard2022approximation},   circuit design \cite{CKRS90, CKRS92}, and distributed algorithms \cite{Awerbuch85, BCLR86}.

There are several ways to formalize the manner in which a spanner or emulator approximately preserves distances in $G$, such as multiplicative spanners \cite{ADDJS93} or sublinear additive emulators \cite{thorup2006spanners}. Perhaps the most optimistic formalization requires that distances in $G$ are preserved up to a small purely additive error term. Emulators with purely additive error are called additive emulators, and they will be the focus of this paper.




\begin{definition} 
For a graph $G = (V, E)$, a graph $H = (V, E', w)$ is a $+k$ additive emulator of $G$ if, for all vertices $s, t$, we have $\dist_G(s, t) \leq \dist_H(s, t) \le \dist_G(s, t) + k$.
\end{definition}

Additive spanners were introduced in \cite{ACIM99}, where it was proved that every $n$-vertex graph admits a $+2$ additive spanner of size $O(n^{3/2})$. Later it was shown in \cite{dor2000all} that $+4$ emulators of size $O(n^{4/3})$ can be obtained. Unfortunately, the existence of polynomially sparser emulators with constant additive error was ruled out by \cite{AB17jacm}, which proved that in general, emulators with $O(n^{4/3 - \epsilon})$ edges suffered $+n^{\Omega(1)}$ additive error.

Consequently, sparse emulators must suffer polynomial additive error. However, it has remained open precisely what polynomial additive error is achievable for emulators of size $O(n^{4/3 - \epsilon})$. A particularly interesting setting is when the emulator $H$ is as sparse as possible, i.e., $H$ is of linear size.  The first linear-size additive emulator was given implicitly in \cite{pettie2007low} with additive error $+\Oish(n^{1/4})$. More recently, the existence of linear-size emulators with $+O(n^{0.222 - o(1)})$ additive error was established in \cite{kogan2023new}. We present a new linear-size emulator construction that achieves additive error roughly $+O(n^{0.191})$.

\begin{theorem}
For any $\epsilon > 0$, every $n$-vertex graph has a $+O(n^{\frac{1}{3 + \sqrt{5}}+ \epsilon})$ additive emulator on $O_{\epsilon}(n)$ edges. 
\end{theorem}
\noindent

Lower bounds on the additive error of linear-size emulators were initiated in \cite{HP18}, which established that $+\Omega(n^{1/18})$ additive error is necessary in general. This was subsequently improved to $+\Omega(n^{2/29})$ in \cite{LVWX22}. However, a significant gap remains between the best known upper and lower bounds on the additive error of linear-size emulators.

\section{Technical Overview}
The prior linear-size emulator construction with $+O(n^{0.222 - o(1)})$  additive error of \cite{kogan2023new} made use of a clever discretization of a weighted-variant of Thorup-Zwick emulators. Since Thorup-Zwick emulators require superlinear space, the construction in \cite{kogan2023new} inserted the modified TZ emulator over a subsampled net of the input graph. We diverge from this approach, instead returning to the clustering and path-buying strategies used to construct additive spanner upper bounds in \cite{BKMP05x, BV16, tan2023almost}.

Our emulator construction begins with  a graph clustering decomposition of \cite{BV16} that is standard in the area. This clustering decomposes the input graph $G$ into a collection of clusters $C_1, \dots, C_k$, each of radius $\leq r$, with certain `nice' covering properties. (See Lemma \ref{lem:clustering_lem} for details.) We  will construct a linear-size emulator $H$ of $G$ with additive error $+r$, where $r>0$ is an integer parameter to be optimized in our construction. 

 Each cluster $C_i$ in our clustering  is categorized as either being \textit{small} if $C_i$ contains fewer than $|C_i| \leq O(r^2)$ vertices or \textit{large} otherwise. We will  handle each cluster $C_i$ based on its classification: 
\begin{itemize}
    \item If cluster $C_i$ is small, then we will exactly preserve  paths passing through the cluster in our emulator $H$ using a simple sampling scheme.
    \item If cluster $C_i$ is large, then we will  approximately preserve paths passing through the cluster in our emulator $H$ by recursively inserting a linear-size emulator of $C_i$ into $H$.
\end{itemize}

The key technical development  that allows us to apply the path-buying method successfully is our sampling scheme for handling small clusters. After adding a small number of edges to $H$, we are roughly able to assume that all balls of radius $r$ in $G$ contain $\Omega(r^2)$ vertices. This property has previously been called quadratic expansion in the context of $(1+\varepsilon, \beta)$ spanners \cite{pettie2007low}. 
The sampling scheme that handles small clusters is based on the following observation.




\begin{observation}
    Let $G$ be an $n$-vertex graph such that every subgraph of $G$ of radius $\leq r$ contains at most $O(r^2)$ vertices. Then $G$ admits an additive emulator with $\Oish(n)$ edges and error $+\Oish(r)$. 
\end{observation}
\begin{proof}[Proof sketch]
    Let $G = (V, E)$. 
    Add a linear-size $\cdot \log n$ multiplicative spanner to our emulator $H$ of $G$. 
    Sample each vertex $v \in V$ into set $V'$ independently with probability $\Theta(\log n / r)$. For all $s, t \in V'$ such that $\dist_G(s, t) \leq r$, add the emulator edge $(s, t)$ with weight $\dist_G(s, t)$ to $H$. This completes the construction.
    
    Size bound: each vertex $v \in V'$ has at most $O(r^2)$ vertices within radius $r$, so at most $O(r \log n)$  edges incident to $v$ are added to $H$ in expectation. Then the total number of edges added to $H$ is $|V'| \cdot O(r \log n) = \Oish(n)$ in expectation. 

    Error bound: fix a pair of vertices $s, t \in V$ and an $s \leadsto t$-path $\pi$. For each subpath $\pi'$ of $\pi$ of length at least $r/3$, set $V'$ hits path $\pi'$ with high probability (i.e., $\pi' \cap V' \neq \emptyset$). Now let $v_1, \dots, v_k$ be the vertices in $\pi \cap V'$, listed in the order they appear in path $\pi$. Then edge $(v_i, v_{i+1})$ is in $H$ for all $i$, with high probability. This implies that $\dist_H(v_1, v_k) = \dist_G(v_1, v_k)$. We conclude:
    $$
    \dist_H(s, t) \leq \dist_H(s, v_1) + \dist_H(v_1, v_k) + \dist_H(v_k, t) \leq \dist_G(v_1, v_k) + \Oish(r).
    \qedhere
    $$
    
\end{proof}

Finally, the remaining edges are added to our emulator in a greedy fashion similar to the path-buying strategies of \cite{BKMP05x, BV16} used in prior additive spanner  constructions. While there exists a pair of vertices $s, t \in V$ that do not satisfy the additive emulator condition in $H$ (e.g.,  $\dist_H(s, t) > \dist_G(s, t) + r$), we will add an emulator edge to $H$ to ensure that the additive emulator condition becomes satisfied between $s$ and $t$. We will then use a simple counting argument to argue that at most $O(n)$ emulator edges are added to $H$ in this way.

\section{Preliminaries}
For the remainder of this paper, we  use $\widehat{O}(\cdot)$ to hide $n^{O(\epsilon)}$ factors and $O_{\epsilon}(\cdot)$ to hide constant factors dependent on $\epsilon$  in our asymptotic notation. We let $B(v, r)$ denote the set of vertices of distance at most $r$ from $v$ in $G$. 

\begin{lemma}[Lemma 13 of \cite{BV16}] 
\label{lem:clustering_lem}
Let  $r \in [1, n]$ and  $\epsilon > 0$. 
For every $n$-vertex graph $G = (V, E)$, there exists a set of vertices $\mathcal{C} = \{v_1, \dots, v_k\}$ and corresponding integers $\mathcal{R} = \{r_1, \dots, r_k\}$, where $r_i = \widehat{\Theta}(r)$, satisfying the following:
\begin{itemize}
    \item (Coverage) For each $v \in V$, $v \in B(v_i, r_i)$ for some $i \in [1, k]$.  
    \item (Low Overlap) $\sum_{i=1}^k|B(v_i, 2r_i)| = O_{\epsilon}(n)$.
\end{itemize}
\end{lemma}
For $i \in [1, k]$, we will think of set $B(v_i, 2r_i)$ as a cluster of vertices in $G$, and we will refer to $B(v_i, r_i)$ as the core of cluster $B(v_i, 2r_i)$. Note that the coverage property of Lemma \ref{lem:clustering_lem} states that every vertex in $V$ belongs to the core of some cluster in the  decomposition. For every vertex $v$ in $V$, we let $C(v)$ denote a cluster $B(v_i, 2r_i)$ containing $v$ in its core.


\section{Construction}
In our construction, we will recursively use old emulator upper bounds to obtain new and improved emulator upper bounds. This is formalized in the following lemma.

\begin{lemma}
Suppose for every $n$-vertex graph $G$ and every $\epsilon > 0$ there exists an additive emulator of $G$ with $O_{\epsilon}(n)$ edges and error $+O(n^{\alpha + \epsilon})$, where $\alpha \in (0, 1)$. Then for every $n$-vertex graph $G$ and every $\epsilon > 0$ there exists an additive emulator of $G$ with $O_{\epsilon}(n)$ edges and error $+O(n^{1/(6-4\alpha) + \epsilon })$. 
\label{lem:emu_recursion}
\end{lemma}

\begin{proof} Let $H$ denote our new emulator of $G$. We first give the construction of $H$, and then we prove that $H$ has our desired properties.

\paragraph{Construction of $H$.}  Begin the preprocessing phase by adding a linear-sized multiplicative spanner with $\cdot \log n$ distortion to $H$. Let $r = n^{1/(6 - 4 \alpha)}$.  Sample each $v \in V$ into $V'$ independently with probability $\Theta(r^{-1} \log n)$. Now perform the clustering decomposition of Lemma \ref{lem:clustering_lem} with parameters $r$ and  $\epsilon > 0$ to obtain a set of cluster centers $\mathcal{C} = \{v_1, \dots, v_k\}$ and corresponding cluster radii $\mathcal{R} = \{r_1, \dots, r_k\}$. By Lemma \ref{lem:clustering_lem},  there exists a universal constant $c > 0$  such that  $4r_i \log n \leq r\cdot n^{c\epsilon}$ for $i \in [1, k]$. We let $\hat{r}$ denote $r\cdot n^{c\epsilon}$. 

Fix an $i \in [1, k]$. We say that $B(v_i, 2r_i)$ is a \textit{small} cluster if $|B(v_i, 2r_i)| \leq  r^{2}\log^{-2}n$. Else we say that $B(v_i, 2r_i)$ is a \textit{large} cluster. If $B(v_i, 2r_i)$ is a small cluster, then for all  vertices $s, t$ in $B(v_i, 2r_i) \cap V'$, add the emulator edge $(s, t)$ with weight $d_G(s, t)$ to $H$.
If $B(v_i, 2r_i)$ is a large cluster, then we recursively call our presupposed emulator procedure on the induced subgraph $G' = G[B(v_i, 2r_i)]$ with parameter $\epsilon$ to obtain a emulator $H'$ with $O_{\epsilon}(|B(v_i, 2r_i)|)$ edges and additive error $+O(|B(v_i, 2r_i)|^{\alpha + \epsilon})$. We add the edges of $H'$ to $H$.
By repeating the previous steps for all $i \in [1, k]$ we complete the preprocessing phase.

Now we greedily add emulator edges to $H$ to connect the remaining pairs of vertices violating our spanner property. 
While there exists vertices $s, t$ in $V$ such that $d_H(s, t) > d_G(s, t) + 16\hat{r}$, do the following. Let $\pi$ be a shortest $(s, t)$-path. Let $x$ be the vertex in $\pi$ farthest from $s$ such that
for any two vertices $u, w$ in $ \pi(s, x)$, $d_H(u, w) \leq d_G(u, w) + \hat{r}$. Additionally,  let $y$ be the vertex in $\pi$ farthest from $t$ such that for any two vertices $u, w$ in $\pi(y, t)$,   $d_H( u, w) \leq d_G(u, w) + \hat{r}$. Add emulator edge $(x, y)$ with weight $d_G(x, y)$ to $H$. This greedy phase completes the construction. Our procedure is summarized in Figure \ref{fig:new_emulator}. 

\begin{figure}[h!]
    \centering

\fbox{
\begin{minipage}[c]{0.9\textwidth}
\begin{center}
    \textbf{Additive Emulator Procedure} \texttt{NewEmulator}:
\end{center}
\textbf{Input:}
\begin{itemize}
    \item An $n$-vertex graph $G = (V, E)$ and a parameter $\epsilon > 0$.
    \item A procedure \texttt{OldEmulator}, which takes as input a graph $G'$, and outputs an additive emulator $H'$ of $G'$ with   $O_{\epsilon}(|G'|)$ edges and error $+O(|G'|^{\alpha + \epsilon})$. 
\end{itemize}
    \textbf{Output:} An additive emulator $H$ of $G$ with $O_{\epsilon}(n)$ edges and error $+\widehat{O}(n^{1/(6 - 4  \alpha ) })$. 
\\

\textbf{Preprocessing Phase:}
\begin{enumerate}
    \item Add a linear-sized spanner with $\cdot \log n$ multiplicative distortion to $H$.
    \item Let $r = n^{1/(6 - 4 \alpha )}$. Sample each vertex $v \in V$ into $V'$ independently with probability $\Theta(\log n / r)$. 
    \item Perform the clustering decomposition of Lemma \ref{lem:clustering_lem} with parameters $r$ and $\epsilon$ to obtain a set of cluster centers $\mathcal{C} = \{v_1, \dots, v_k\}$ and corresponding cluster radii $\mathcal{R} = \{r_1, \dots, r_k\}$. Let $\hat{r} = r \cdot n^{c\epsilon}$ for a sufficiently large universal constant $c$.
    \item For each $i \in [1, k]$: 
    \vspace{2mm}   \\
    $\bullet$  If $|B(v_i,  2r_i)| \leq r^2 \log^{-2} n$, then for all vertices $s, t$ in $B(v_i, 2r_i) \cap V'$, add  edge $(s, t)$ with weight $d_G(s, t)$ to $H$.
    \vspace{2mm}   \\
    $\bullet$ Else if $|B(v_i, 2r_i)| > r^{2}\log^{-2}n$, then call \texttt{OldEmulator} on graph $G[B(v_i, 2r_i)]$. Add the emulator output by \texttt{OldEmulator} to $H$.
\end{enumerate}
\textbf{Greedy Phase:}

While there exist vertices $s, t$ in $V$ such that $d_H(s, t) > d_G(s, t) + 16\hat{r}$, do the following. Let $\pi$ be a shortest $(s, t)$-path. 
Let $x$ be the vertex in $\pi$ farthest from $s$ such that
for any two vertices $u, w$ in $ \pi(s, x)$, $d_H(u, w) \leq d_G(u, w) + \hat{r}$. Additionally,  let $y$ be the vertex in $\pi$ farthest from $t$ such that for any two vertices $u, w$ in $\pi(y, t)$,   $d_H( u, w) \leq d_G(u, w) + \hat{r}$. Add edge $(x, y)$  with weight $d_G(x, y)$ to $H$. 
\end{minipage}
}
    \caption{The recursive procedure for the improved emulator upper bounds. }
    
    \label{fig:new_emulator}
\end{figure}

\paragraph{Bounding the error of $H$.}  Fix  vertices $s, t \in V$.  If $d_H(s, t) \leq d_G(s, t) + 16\hat{r}$, then we are done. 
Otherwise, vertex pair $(s, t)$ is considered in some round of the greedy phase. Let path $\pi$ and vertices $x, y$ be as defined in the construction. Note that since edge $(x, y)$ is in $H$, $d_H(x, y) = d_G(x, y)$. 
We have the following:
\begin{align*}
    d_H(s, t) & \leq d_H(s, x) + d_H(x, y) + d_H(y, t) \\ 
   & \leq (d_G(s, x) + \hat{r}) + d_G(x, y) + d_G(y, t) + \hat{r}) \\
 &   \leq d_G(s, t) + 2\hat{r}
\end{align*}
Since $\hat{r} = r \cdot  n^{c\varepsilon}$ for a universal constant $c > 0$, we may obtain our desired error by taking our
construction parameter $\varepsilon > 0$ to be sufficiently small.

\paragraph{Bounding the size of $H$.} We begin by bounding the number of edges added to $H$ in the preprocessing phase.  Fix an $i \in [1, k]$. If $B(v_i, 2r_i)$ is a small cluster, then with high probability $$|B(v_i, 2r_i) \cap V'| = \Theta(|B(v_i, 2r_i)| \cdot r^{-1} \log n).$$ Since we add an emulator edge between every pair of vertices in $B(v_i, 2r_i) \cap V'$, it follows that we add at most $$|B(v_i, 2r_i) \cap V'|^2 =  \Theta(|B(v_i, 2r_i)|^2 \cdot r^{-2} \log^2 n) = O(|B(v_i, 2r_i)|)$$ edges to $H$. Otherwise, if $B(v_i, 2r_i)$ is a large cluster, then our recursive emulator call contributes $O_{\epsilon}(|B(v_i, 2r_i)|)$ edges to $H$. Then by the low overlap property of Lemma \ref{lem:clustering_lem}, we add $O_{\epsilon}(n)$ edges to $H$ in the preprocessing phase.

To bound the number of edges added in the greedy phase, we use a path buying argument reminiscent of the proof in \cite{BV16}. We say that vertices $s, t$ in $V$ are \textit{connected} in $H$ if $d_H(s, t) \leq d_G(s, t) + 8\hat{r}$. Each time we add an emulator edge to $H$ in the greedy phase, we will argue that $\Omega(n)$ pairs of vertices in $V$ become connected in $H$ for the first time. Then since there are $O(n^2)$ pairs of vertices in $V$, the greedy phase adds $O(n)$ edges to $H$. 

Let $(s, t)$ be a pair of vertices considered in some round of the greedy phase. Let path $\pi$ and vertices $x, y$ be as defined in the construction. We say that paths $\pi(s, x)$ and $\pi(y, t)$ are the prefix and suffix of $\pi$, respectively. We will define a set $S$ of vertices in $G$ near the prefix of $\pi$ and a set $T$ of vertices in $G$ near the suffix of $\pi$. Then we will show that after adding edge $(x, y)$ to $H$,  all pairs of vertices in $S \times T$ become connected in $H$ for the first time. Finally, we will establish that $|S \times T| = \Omega(n)$, completing the proof.

We define $S$ and $T$ as follows:
\[
S := \cup_{v \in \pi(s, x)} C(v) \qquad \text{ and } \qquad T := \cup_{v \in \pi(y, t)} C(v).
\]

We now verify that $S$ and $T$ satisfy our desired properties.

\begin{proposition}
After edge $(x, y)$ is added to $H$, all pairs of vertices in $S \times T$ become connected in $H$ for the first time.
\label{prop:first_conn}
\end{proposition}
\begin{proof}
Let $v \in V$, and suppose that $v$ is contained in the core of a cluster $C(v)$ with radius $2r_i$, where $r_i \in \mathcal{R}$. Let $u \in C(v)$. Then by our choice of $\hat{r}$,
$$
\dist_G(v, u) \leq 4r_i \leq \hat{r} \cdot \log^{-1} n.
$$
Therefore, for each vertex $u \in S$, there exists a vertex $v \in \pi(s, x)$ such that 
$$\dist_H(v, u) \leq \log n \cdot \dist_G(v, u) \leq \hat{r},$$
where the first inequality follows from the $\cdot \log n$ multiplicative spanner in $H$. 
By an identical argument, for each vertex $ u \in T$, there exists a vertex $v \in \pi(y, t)$ such that
$$\dist_H(v, u) \leq \log n \cdot \dist_G(v, u) \leq \hat{r}.$$

Now fix vertices $s' \in S$ and $t' \in T$, and let $x'$  (respectively, $y'$) be the vertex in $\pi(s, x)$ (respectively,  $\pi(y, t)$) that is closest to $s'$ (respectively,  $t'$) in $G$.
 (See Figure \ref{fig:emulator_pic} for a visualization of this situation.) Note that by the above argument, $\dist_H(s', x') \leq \hat{r}$ and $\dist_H(y', t') \leq \hat{r}$. 
After edge $(x, y)$ is added to $H$, $s'$ and $t'$ are connected in $H$:
\begin{align*}
d_H(s', t')  & \leq d_H(s', x') + d_H(x', y') + d_H(y', t') \\
& \leq  \hat{r} + d_H(x', y') + \hat{r}  \\
& \leq  d_H(x', x) + d_H(x, y) + d_H(y, y') + 2\hat{r} \\
& \leq  (d_G(x', x) + \hat{r}) + d_G(x, y) + (d_G(y, y') + \hat{r}) +2\hat{r} \\
& \leq  d_G(x', y')  + 4\hat{r} \\
& \leq d_G(x', s') + d_G(s', t') + d_G(t', y') + 4\hat{r} \\
& \leq d_G(s', t') + 6\hat{r},
\end{align*}
where the fourth inequality follows from our choice of $x$ and $y$, and the fifth inequality follows from the fact that $\pi(x', y')$ is a shortest path in $G$.

Now suppose for the sake of contradiction that $s'$ and $t'$ were connected in $H$ \textit{before} edge $(x, y)$ is added to $H$. Then we claim that pair $(s, t)$ had low additive error in $H$ before edge $(x, y)$ was added: 
\begin{align*}
    d_H(s, t) & \leq d_H(s, x') + d_H(x', y') + d_H(y', t) \\
    & \leq (d_G(s, x') + \hat{r}) + d_H(x', s') + d_H(s', t') + d_H(t', y') + (d_G(y', t) + \hat{r})  \\
    & \leq (d_G(s, x')+\hat{r}) + \hat{r} + d_H(s', t') + \hat{r} + (d_G(y', t) + \hat{r}) \\
    & \leq d_G(s, x') + d_H(s', t') + d_G(y', t) + 4\hat{r} \\
    & \leq d_G(s, x') + (d_G(s', t') + 8\hat{r}) + d_G(y', t) + 4\hat{r} \\
    & \leq d_G(s, x') + (d_G(s', x') + d_G(x', y') + d_G(y', t')) + d_G(y', t) + 12\hat{r} \\
    & \leq d_G(s, x') + d_G(x', y') + d_G(y', t) + 14\hat{r} \\
    & \leq d_G(s, t) + 14\hat{r}
\end{align*}
This contradicts our assumption that pair $(s, t)$ was considered in this round of the greedy phase. We conclude that all pairs of vertices in $S \times T$ are connected for the first time after edge $(x, y)$ is added to $H$.
\end{proof}

\begin{figure}[h]
    \centering
    \includegraphics{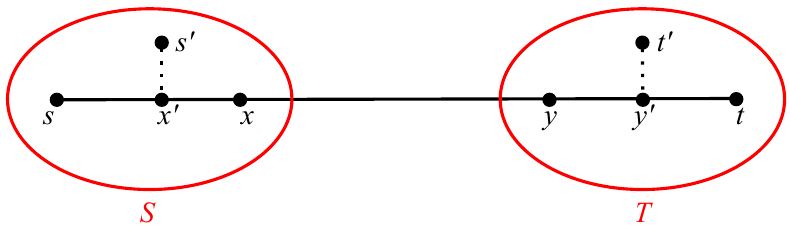}
    \caption{The solid black line denotes path $\pi$ in $G$. After edge $(x, y)$ is added to $H$ in the greedy phase, all pairs of vertices  $(s', t')$ in $S \times T$ become connected in $H$ for the first time. }
    \label{fig:emulator_pic}
\end{figure}


What remains is to show that $|S \times T| = \Omega(n)$. Specifically, we will show that $|S| = \Omega(n^{1/2})$, and $|T| = \Omega(n^{1/2})$ will follow by a symmetric argument. Notice that the size of $S$ is constrained by our requirement that for any two vertices $u, w$ in $\pi(s, x)$, $d_H(u, w) \leq d_G(u, w) + \hat{r}$. 
\begin{proposition}
$|S \times T| = \Omega(n)$
\label{prop:s_t_size}
\end{proposition}
\begin{proof}
Instead of lower bounding $|S|$ directly, we  find it easier to lower bound a subset $S' \subseteq S$. We now explicitly construct $S'$. We begin the construction by initializing two sets as $U := \emptyset$ and $S' := \emptyset$ and setting an integer counter 
$\texttt{err} := 0$. Starting at $s$, we walk through $\pi(s, x)$ as follows. Let $s_1$ be a vertex in $\pi(s, x) \cap V'$ such that $d_G(s, s_1) \leq r$. Note that since we sample vertices in $V$ into $V'$ with probability $\Theta(\log n / r)$, such a vertex $s_1$ will exist with high probability by the Chernoff bound. Add $s_1$ to $U$ and add the vertices in $C(s_1)$ to $S'$. Now, given $s_i$, we let $s_{i+1}$ be a vertex in $\pi(s_i, x) \cap V'$ such that $r/2 \leq d_G(s_i, s_{i+1}) \leq r$. Again by the Chernoff bound, such a vertex $s_{i+1}$ will exist with high probability.
Add vertex $s_{i+1}$ to $U$, and add the vertices in cluster $C(s_{i+1})$ to $S'$. 
If $C(s_i)$ is a large cluster, then we increment $\texttt{err}$ by $+O(|C(s_i)|^{\alpha + \epsilon})$, which corresponds to the additive error of the linear sized emulator of $C(s_i)$ we inserted into $H$ in the preprocessing phase. We proceed in this manner until we add the first vertex $s_{\ell}$ such that $\texttt{err} + O(|C(s_{\ell})|^{\alpha + \epsilon}) > \hat{r} / 2$. (This implies that  $\texttt{err} \leq \hat{r} /2$.) Note that $S' = \cup_{s_i \in U}C(s_i)$.

Now we verify that $S' \subseteq S$, as desired. To prove this, it will suffice to show that $\pi(s, s_{\ell})$ is a subpath of $\pi(s, x)$. Fix an $i \in [1, \ell - 1]$ and suppose that $C(s_i)$ is a small cluster. Then since $s_i$ is in the core of cluster $C(s_i)$ and $d_G(s_i, s_{i+1}) \leq r$, it follows that $s_{i+1} \in C(s_i)$.  Note that in this case, emulator edge $(s_i, s_{i+1})$ is added to $H$ in the preprocessing phase, allowing us to travel from $s_i$ to $s_{i+1}$ without incurring any error. Specifically, $d_H(s_i, s_{i+1}) = d_G(s_i, s_{i+1})$.  Otherwise, $C(s_i)$ is a large cluster. In this case, we added an emulator of $G[C(s_i)]$ with additive error $+O(|C(s_i)|^{\alpha + \epsilon})$ to $H$ in the preprocessing phase. Note that since $s_i$ is in the core of $C(s_i)$ and $d_G(s_i, s_{i+1}) \leq r$, it follows that any shortest $(s_i, s_{i+1})$-path in $G$ is contained in $G[C(s_i)]$. Consequently, $d_H(s_i, s_{i+1}) \leq d_G(s_i, s_{i+1}) + O(|C(s_i)|^{\alpha + \epsilon})$.
Then by the previous observations, we  conclude that $d_H(s_i, s_j) \leq d_G(s_i, s_j) + \texttt{err}$ for all $i, j \in [1, \ell]$. 
Now let $u, w$ be vertices in $\pi(s, s_{\ell})$ (where $u$ occurs before $w$ in $\pi(s, s_{\ell})$), and let $s_i, s_j$ be the vertices in $\pi(s, s_{\ell}) \cap U$ that are closest to  $u$ and closest to $w$, respectively. Observe that
\begin{align*}
    d_H(u, w) & \leq d_H(u, s_i) + d_H(s_i, s_j) + d_H(s_j, w) \\
    & \leq \log n \cdot d_G(u, s_i) + d_H(s_i, s_j) + \log n \cdot d_G(s_j, w) \\
    & \leq d_G(u, s_i) + d_H(s_i, s_j) + d_G(s_j, w) + 2r\log n \\
    & \leq d_G(u, s_i) + (d_G(s_i, s_j) + \texttt{err}) + d_G(s_j, w) + 2r\log n \\
    & \leq d_G(u, w) + \hat{r}/2 + 2r \log n \\
    & \leq d_G(u, w) + \hat{r}.
\end{align*}
Then by our choice of $x$, it follows that $\pi(s, s_{\ell})$ is a subpath of $\pi(s, x)$, so $S' \subseteq S$. 

By our construction of $S'$, we have the guarantee that $\sum_{i=1}^{\ell} O(|C(s_{i})|^{\alpha + \epsilon}) > \hat{r} / 2$. 
Moreover, we may assume that all clusters $C(s_i)$ are  large clusters, since as noted earlier we can travel through small clusters with zero error. This implies that $|C(s_i)| \geq r^2 \log^{-2} n$. 
Now we wish to lower bound the sum $\sum_{i=1}^{\ell} |C(s_i)|$. We may assume that $\alpha + \epsilon < 1$, and so by a convexity argument it can be seen that the sum $\sum_{i=1}^{\ell} |C(s_i)|$ is minimized (while subject to our guaranteed  inequality) when each large cluster has size $r^2 \log^{-2}n$. Let $q$ be the number of large clusters we pass through in $\pi(s, s_{\ell})$.  Then our guarantee becomes  $\sum_{i=1}^{\ell} O(|C(s_{i})|^{\alpha + \epsilon}) = q \cdot (r^2 \log^{-2}n)^{\alpha + \epsilon} > \hat{r} / 2$. Rearranging this inequality gives us that $q > \hat{r} \cdot r^{-2(\alpha + \epsilon)}$. Now note that since $\hat{r} = r \cdot n^{c\epsilon}$ for a sufficiently large constant $c$, we may assume that $\hat{r} \geq r^{1 + 2\epsilon} \log^2 n$ by taking $c > 2$. Now putting it all together, we see that
\[
\sum_{i=1}^{\ell}|C(s_i)| \geq q \cdot r^2 \log^{-2} n \geq \frac{\hat{r} \cdot r^2}{r^{2(\alpha + \epsilon)} \cdot \log^2 n} \geq r^{3 - 2\alpha} = n^{(3 - 2\alpha)/(6 - 4\alpha)} = \Theta(n^{1/2}),
\]
where the second to last equality follows from our choice of $r = n^{1/(6 - 4\alpha)}$. 

Now we finish the proof by showing that $|S'| = |\cup_{i=1}^{\ell} C(s_i)| = \Omega(\sum_{i=1}^{\ell}|C(s_i)|)$. Recall that $d_G(s_i, s_{i+1}) \geq r/2$. Then each vertex $v \in V$ can only occur in $O(1)$ distinct clusters $C(s_i)$, where $s_i \in U$, or else we contradict our assumption that $\pi(s, x)$ is a shortest path in $G$. We conclude that $|\cup_{i=1}^{\ell} C(s_i)| = \Omega(\sum_{i=1}^{\ell}|C(s_i)|)$, so $|S| \geq |S'| = \Omega(n^{1/2})$. A symmetric argument shows that $|T| = \Omega(n^{1/2})$, so it follows that $|S \times T| = \Omega(n)$. 
\end{proof}

 Each time we add an emulator edge to $H$ in the greedy phase of the construction,  $|S \times T| = \Omega(n)$ pairs of vertices in $V$ become connected in $H$ for the first time by Propositions \ref{prop:first_conn} and \ref{prop:s_t_size}. Then since there are $O(n^2)$ pairs of vertices in $V$, the greedy phase adds $O(n)$ edges to $H$. We conclude that we add $O(n)$ emulator edges to $H$ in the greedy phase, so the total size of $H$ is $O_{\epsilon}(n)$. 
In our construction we made use of the Chernoff bound only polynomially many times, so by the union bound our construction succeeds with high probability.  This completes the proof of Lemma \ref{lem:emu_recursion}. 
\end{proof}

Now we can repeatedly apply Lemma \ref{lem:emu_recursion} to obtain a sequence of improved emulator upper bounds. We choose our initial emulator to be a spanning tree with error $+n$ and $n-1$ edges.
Then after one application of Lemma \ref{lem:emu_recursion}, we obtain an additive emulator with $O_{\epsilon}(n)$ edges and error $+O(n^{1/(6 - 4) + \epsilon}) = +O(n^{1/2 + \epsilon})$. More generally, after $i$ applications of Lemma \ref{lem:emu_recursion}, we obtain an additive emulator with $O_{\epsilon}(n)$ edges and error $+O(n^{a_i + \epsilon})$, where $a_i$ is defined by the recurrence relation $a_0 = 1$, $a_{i+1} = 1/(6-4a_i)$. The value of $a_i$ converges to the fixed point $\frac{1}{3 + \sqrt{5}} \approx 0.191$.

\begin{theorem}
For every $n$-vertex graph $G = (V, E)$ and $\epsilon > 0$, there exists an emulator $H = (V, E')$ on $O_{\epsilon}(n)$ edges with error $+O(n^{\frac{1}{3 + \sqrt{5}}+ \epsilon})$. 
\end{theorem}

We note that the dependency on $\varepsilon$ in the size of $H$ is roughly $O(1)^{1/\varepsilon}$ due to the clustering decomposition of \cite{BV16} used in the preprocessing phase.

\section{Acknowledgements}
I am very grateful to my advisor, Greg Bodwin, for many helpful discussions and for reviewing this paper.

\bibliographystyle{plain}
\bibliography{refs}

\begin{thebibliography}{10}

\bibitem{AB17jacm}
Amir Abboud and Greg Bodwin.
\newblock The 4/3 additive spanner exponent is tight.
\newblock {\em Journal of the ACM (JACM)}, 64(4):28:1--28:14, 2017.

\bibitem{ACIM99}
Donald Aingworth, Chandra Chekuri, Piotr Indyk, and Rajeev Motwani.
\newblock Fast estimation of diameter and shortest paths (without matrix
  multiplication).
\newblock {\em SIAM Journal on Computing}, 28(4):1167--1181, 1999.

\bibitem{ADDJS93}
Ingo Alth{\"o}fer, Gautam Das, David Dobkin, Deborah Joseph, and Jos{\'e}
  Soares.
\newblock On sparse spanners of weighted graphs.
\newblock {\em Discrete \& Computational Geometry}, 9(1):81--100, 1993.

\bibitem{Awerbuch85}
Baruch Awerbuch.
\newblock Complexity of network synchronization.
\newblock {\em Journal of the ACM (JACM)}, 32(4):804--823, 1985.

\bibitem{BKMP05x}
Surender Baswana, Telikepalli Kavitha, Kurt Mehlhorn, and Seth Pettie.
\newblock New constructions of ($\alpha$, $\beta$)-spanners and purely additive
  spanners.
\newblock In {\em Proceedings of the 16th Annual ACM-SIAM symposium on Discrete
  algorithms (SODA)}, pages 672--681. Society for Industrial and Applied
  Mathematics, 2005.

\bibitem{BCLR86}
Sandeep Bhatt, Fan Chung, Tom Leighton, and Arnold Rosenberg.
\newblock Optimal simulations of tree machines.
\newblock In {\em 27th Annual Symposium on Foundations of Computer Science
  (sfcs 1986)}, pages 274--282. IEEE, 1986.

\bibitem{BV16}
Greg Bodwin and Virginia~Vassilevska Williams.
\newblock Better distance preservers and additive spanners.
\newblock {\em ACM Trans. Algorithms}, 17(4), oct 2021.

\bibitem{chen2022maximum}
Li~Chen, Rasmus Kyng, Yang~P Liu, Richard Peng, Maximilian~Probst Gutenberg,
  and Sushant Sachdeva.
\newblock Maximum flow and minimum-cost flow in almost-linear time.
\newblock In {\em 2022 IEEE 63rd Annual Symposium on Foundations of Computer
  Science (FOCS)}, pages 612--623. IEEE, 2022.

\bibitem{CKRS90}
Jason Cong, Andrew~B Kahng, Gabriel Robins, Majid Sarrafzadeh, and CK~Wong.
\newblock {\em Performance-driven global routing for cell based IC's}.
\newblock University of California (Los Angeles). Computer Science Department,
  1990.

\bibitem{CKRS92}
Jason Cong, Andrew~B Kahng, Gabriel Robins, Majid Sarrafzadeh, and CK~Wong.
\newblock Provably good algorithms for performance-driven global routing.
\newblock {\em Proc. of 5th ISCAS}, pages 2240--2243, 1992.

\bibitem{dalirrooyfard2022approximation}
Mina Dalirrooyfard, Ce~Jin, Virginia~Vassilevska Williams, and Nicole Wein.
\newblock Approximation algorithms and hardness for n-pairs shortest paths and
  all-nodes shortest cycles.
\newblock In {\em 2022 IEEE 63rd Annual Symposium on Foundations of Computer
  Science (FOCS)}, pages 290--300. IEEE, 2022.

\bibitem{dor2000all}
Dorit Dor, Shay Halperin, and Uri Zwick.
\newblock All-pairs almost shortest paths.
\newblock {\em SIAM Journal on Computing}, 29(5):1740--1759, 2000.

\bibitem{HP18}
Shang-En Huang and Seth Pettie.
\newblock {Lower Bounds on Sparse Spanners, Emulators, and Diameter-reducing
  shortcuts}.
\newblock In David Eppstein, editor, {\em 16th Scandinavian Symposium and
  Workshops on Algorithm Theory (SWAT 2018)}, volume 101 of {\em Leibniz
  International Proceedings in Informatics (LIPIcs)}, pages 26:1--26:12,
  Dagstuhl, Germany, 2018. Schloss Dagstuhl--Leibniz-Zentrum fuer Informatik.

\bibitem{kogan2023new}
Shimon Kogan and Merav Parter.
\newblock New additive emulators.
\newblock In {\em 50th International Colloquium on Automata, Languages, and
  Programming (ICALP 2023)}. Schloss Dagstuhl-Leibniz-Zentrum f{\"u}r
  Informatik, 2023.

\bibitem{LVWX22}
Kevin Lu, Virginia~Vassilevska Williams, Nicole Wein, and Zixuan Xu.
\newblock Better lower bounds for shortcut sets and additive spanners via an
  improved alternation product.
\newblock In {\em Proceedings of the 2022 Annual ACM-SIAM Symposium on Discrete
  Algorithms (SODA)}, pages 3311--3331. SIAM, 2022.

\bibitem{pettie2007low}
Seth Pettie.
\newblock Low distortion spanners.
\newblock In {\em Automata, Languages and Programming: 34th International
  Colloquium, ICALP 2007, Wroc{\l}aw, Poland, July 9-13, 2007. Proceedings 34},
  pages 78--89. Springer, 2007.

\bibitem{tan2023almost}
Zihan Tan and Tianyi Zhang.
\newblock Almost-optimal sublinear additive spanners.
\newblock In {\em Proceedings of the 55th Annual ACM Symposium on Theory of
  Computing}, pages 281--294, 2023.

\bibitem{thorup2006spanners}
Mikkel Thorup and Uri Zwick.
\newblock Spanners and emulators with sublinear distance errors.
\newblock In {\em Proceedings of the seventeenth annual ACM-SIAM symposium on
  Discrete algorithm}, pages 802--809, 2006.

\end{thebibliography}

\end{document}